\documentclass[conference]{IEEEtran}
\usepackage{booktabs} 

\usepackage[utf8]{inputenc}
\usepackage{xspace}
\usepackage{mathrsfs}
\usepackage{amsmath}
\usepackage{amssymb}
\usepackage{amsfonts}
\usepackage{times}
\usepackage{graphicx}
\usepackage{latexsym}
\usepackage{bm}
\usepackage{float}
\usepackage{url}

\makeatletter{}\usepackage{import}

\usepackage{ifthen}
\usepackage{url}
\usepackage{amssymb}
\usepackage{amsmath}
\usepackage{import}
\usepackage{xparse}
\usepackage{amsmath}

\usepackage{xspace}
\usepackage{datatool}
\usepackage{glossaries}
\usepackage{pbox}
\usepackage{centernot}
\usepackage{mathrsfs}

\usepackage{amsthm}

\newcommand{\m}[1]{\ensuremath{#1}\xspace}
\newcommand{\trval}[1]{\m{\mbox{\bf #1}}}

\newtheorem{definition}{Definition}
\newtheorem{example}{Example}
\newtheorem{lemma}{Lemma}
\newtheorem{theorem}{Theorem}

\usepackage{tikz}

\definecolor{lavander}{cmyk}{0,0.48,0,0}
\definecolor{violet}{cmyk}{0.79,0.88,0,0}
\definecolor{burntorange}{cmyk}{0,0.52,1,0}

\def\lav{black!90}
\def\oran{black!30}

\tikzstyle{peers}=[draw,circle,violet,bottom color=\lav,
                  top color= white, text=violet,minimum width=10pt]
\tikzstyle{superpeers}=[draw,circle,burntorange, left color=\oran,
                       text=violet,minimum width=20pt]
\tikzstyle{legendsp}=[rectangle, draw, burntorange, rounded corners,
                     thin,bottom color=\oran, top color=white,
                     text=burntorange, minimum width=2.5cm]
\tikzstyle{legendp}=[rectangle, draw, violet, rounded corners, thin,
                     bottom color=\lav, top color= white,
                     text= violet, minimum width= 2.5cm]
\tikzstyle{legend_general}=[rectangle, rounded corners, thin,
                           burntorange, fill= white, draw, text=violet,
                           minimum width=2.5cm, minimum height=0.8cm]

	\newcommand{\lrule}{\m{\leftarrow}}
	\newcommand{\cause}{\m{\stackrel{c}{\lrule}}}
	
	\newcommand{\ltrue}{\trval{t}}
	\newcommand{\lfalse}{\trval{f}}
	\newcommand{\lunkn}{\trval{u}}
	
	\newcommand{\Tr}{\ltrue}
	\newcommand{\Fa}{\lfalse}
	\newcommand{\Un}{\lunkn}

	\newcommand{\struct}{\m{I}}
	
	\newcommand{\I}{\m{\mathcal{I}}}

	\NewDocumentCommand\inter{g+g}{%
	  \IfNoValueTF{#1}
	    {\struct}
	    {\m{#1^{#2}}}}

	\newcommand{\ttt}{\m{\overline{t}}}

	\renewcommand{\int}{\m{\mathbb{Z}}}

	\NewDocumentCommand\subs{g+g}{%
	  \IfNoValueTF{#1}
	    {\m{/}}
	    {\m{#1/ #2}}}

	\newcommand{\logicname}[1]{\text{\sc #1}\xspace}

	\newcommand{\idp}{\logicname{IDP}}

	\newcommand{\fodot}{\logicname{FO(\ensuremath{\cdot})}}
	
	\newcommand{\foid}{\logicname{FO(\ensuremath{ID})}}
	
	\newcommand{\fo}{\FO}

\newcommand{\ouracronym}[3]{%
	\newacronym{#1}{#2}{#3}
	\expandafter\newcommand\csname #1\endcsname{\gls{#1}\xspace}%
}
	\ouracronym{FO}{FO}{first-order logic}
	\ouracronym{MX}{MX}{Model Expansion}
	\ouracronym{MO}{MO}{Model Optimization}
	\ouracronym{ASP}{ASP}{Answer Set Programming}
	\ouracronym{TP}{TP}{Theorem Proving}
	\ouracronym{LP}{LP}{Logic Programming}
	\ouracronym{CP}{CP}{Constraint Programming}
	\ouracronym{KR}{KR}{Knowledge Representation}
	\ouracronym{CSP}{CSP}{Constraint Satisfaction Problem}
	\ouracronym{SMT}{SMT}{SAT Modulo Theories}
	\ouracronym{KBS}{KBS}{knowledge-base system}
	\ouracronym{NNF}{NNF}{Negation Normal Form}
	\ouracronym{FNNF}{FNNF}{Flat Negation Normal Form}
	\ouracronym{DefNNF}{DefNNF}{Definition Negation Normal Form}
	\ouracronym{CDCL}{CDCL}{Conflict-Driven Clause-Learning}
	\ouracronym{LCG}{LCG}{Lazy Clause Generation}

	\makeatletter
	\def\ifenv#1{
	\def\@tempa{#1}%
	\def\@ttempa{#1*}%
	\ifx\@tempa\@currenvir
	\expandafter\@firstoftwo
	\else
	\expandafter\@secondoftwo
	\fi
	}
	\makeatother

	\newcommand{\ddrule}[4]{\ensuremath{#1 \leftarrow #2 & \{#3\} & #4}}
	\newcommand{\drule}[2]{\ensuremath{#1 & \leftarrow & #2}}

	\newcommand{\darule}[4]{\ensuremath{#1 \leftarrow #2 & \{#3\} & #4}}
	\newcommand{\arule}[2]{\ensuremath{#1 \, &\leftarrow \, #2}}

	\newcommand{\LNDRule}[2]{
	\ifenv{array}
	{\drule{#1}{#2}}
	{ \ifenv{align}
		{\arule{#1}{#2}}
		{\ifenv{align*}
		{\arule{#1}{#2}}
		{ERROR: using LDRule in unsupported environment: \@currenvir}
		}
	}
	}

	\newcommand{\LDRule}[4]{
	\ifenv{array}
	{\ddrule{#1}{#2}{#3}{#4}}
	{ \ifenv{align}
		{\darule{#1}{#2}{#3}{#4}}
		{\ifenv{align*}
		{\darule{#1}{#2}{#3}{#4}}
		{ERROR: using LDRule in unsupported environment: \@currenvir}
		}
	}
	}

	\NewDocumentCommand\LRule{m+g+g+g}{%
		\IfNoValueTF{#2}%
		{#1.&}{%
		\IfNoValueTF{#3}
		{\LNDRule{#1}{#2}}
		{\LDRule{#1}{#2}{#3}{#4}}%
		}
	}

	\NewDocumentCommand\CLRule{m+g}{%
	\ifenv{array}
	{\cdrule{#1}{#2}}
	{ \ifenv{align}
		{\carule{#1}{#2}}
		{\ifenv{align*}
			{\carule{#1}{#2}}
			{ERROR: using CLRule in unsupported environment: \@currenvir}
		}
	}
	}

	\NewDocumentCommand\carule{m+g}{%
		\IfNoValueTF{#2}
			{\ensuremath{#1.}}
			{\ensuremath{#1 \, &\cause \, #2}}}
	\NewDocumentCommand\cdrule{m+g}{%
		\IfNoValueTF{#2}
			{\ensuremath{#1.}}
			{\ensuremath{#1 & \cause & #2}}}

	\newcommand{\algrule}[4]{
	\hbox{{#1}:}&
	\quad #2 ~\longrightarrow~ #3
	\hbox{~ if } #4\\
	}

	\newcommand{\AlgoRule}[4]{
	\ifenv{array}
	{\algrule{#1}{#2}{#3}{#4}}
		{ERROR: using AlgoRule in unsupported environment: \@currenvir}
	}

\newcommand{\commentstyle}{\color{Gray}}

	\usepackage{listings}

	\lstdefinelanguage{idp}{
		morekeywords=[1]{namespace,vocabulary,theory,structure,procedure,term,set,formula, spec, specification},
		morekeywords=[2]{include,using,type,isa,contains,partial,extern,LFD,GFD,constructed,from,constraint,func,pred,supertype,of,subtype,define},
		morekeywords=[3]{int,float,char,string,nat},
		morekeywords=[4]{if,then,else,for,end},
		morecomment=[s]{/*}{*/},	
		morecomment=[l]{//}
	}
	\newcommand{\ignore}[1]{}

	\newboolean{nocomments}
	\setboolean{nocomments}{false}

	\newboolean{commentmargin}
	\setboolean{commentmargin}{true}

	\newcommand{\namedcomment}[3]{
		\ifthenelse{\boolean{nocomments}}
		{} 
		{ 
			\ifthenelse{\boolean{commentmargin}}
				{ {\color{#3} \marginpar{\color{#3}\sc #2}#1}  } 
				{  {\color{#3} {\sc #2}: #1}  } 
		}
	}
	\newcommand{\mnamedcomment}[3]{\ifthenelse{\boolean{nocomments}}{}{{\marginpar{ \color{#3}{\sc #2}:#1}}}}

	\usepackage{soul}

\newcommand{\keyword}[2]{%
	\expandafter\newcommand\csname #1\endcsname{#2\xspace}%
	\expandafter\newcommand\csname #1s\endcsname{#2s\xspace}%
	\expandafter\newcommand\csname #1ness\endcsname{#2ness\xspace}%
}

\usepackage{tikz}
\usetikzlibrary{arrows}
\usetikzlibrary{positioning}

\usepackage{algorithm}
\usepackage[noend]{algorithmic}

\newcommand{\mc}[1]{\mathcal{#1}}

\newcommand{\T}{\m{\mathbb{T}}} 
\renewcommand{\I}{I}
\newcommand{\Q}{\m{\alpha}} 

\newcommand{\A}{\mc{A}} 

\newcommand{\pws}{\m{Q}} 
\newcommand{\upws}{\m{\mc{Q}}} 

\newcommand{\ubp}{\m{\mc{B}}} 

\newcommand{\leqpws}{\leq_K}
\newcommand{\lequpws}{\leq_K}

\newcommand{\ubprevision}{\m{\mc{D}^*_{\T}}}

\newcommand{\ubprevisionL}{\m{\mc{D}^c_{\T}}}
\newcommand{\ubprevisionU}{\m{\mc{D}^l_{\T}}}

\newcommand{\says}{\,\m{\mathit{says}}\,}
\newcommand{\say}{\m{\mathit{says}}}
\newcommand{\access}{\m{\mathit{access}}}

\newcommand{\DACL}{dAEL(ID)\xspace}

\newcommand{\defs}{\textit{Def}}
\newcommand{\pars}{\textit{Par}}
\newcommand{\wfm}{\textit{wfm}}
\newcommand{\glb}{\textit{glb}}
\newcommand{\LL}{\mathbb{L}}

\renewcommand{\phi}{\varphi}

\ifCLASSINFOpdf
\else
\fi

\begin{document}
%
\title{A Query-Driven Decision Procedure for Distributed Autoepistemic Logic with Inductive Definitions}

\author{\IEEEauthorblockN{Diego Agustin Ambrossio}
\IEEEauthorblockA{University of Luxembourg\\
2, avenue de l'Université\\
L-4365 Esch-sur-Alzette\\
Luxembourg\\
diego.ambrossio@uni.lu}
\and
\IEEEauthorblockN{Marcos Cramer}
\IEEEauthorblockA{University of Luxembourg\\
2, avenue de l'Université\\
L-4365 Esch-sur-Alzette\\
Luxembourg\\
marcos.cramer@uni.lu}}


%


\maketitle

\begin{abstract}
Distributed Autoepistemic Logic with Inductive Definitions (dAEL(ID)) is a recently proposed non-monotonic logic for says-based access control. We define a query-driven decision procedure for dAEL(ID) that is implemented in the knowledge-base system IDP. The decision procedure is designed in such a way that it allows one to determine access rights while avoiding redundant information flow between principals in order to enhance security and reduce privacy concerns. Given that the decision procedure has in the worst case an exponential runtime, it is to be regarded as a proof of concept that increases our understanding of dAEL(ID), rather than being deployed for an access control system.
\end{abstract}


\section{Introduction}

Multiple logics have been proposed for distributed access control \cite{Abadi03,Gurevich07,Abadi08,Garg12,Genovese12}, most of which use a modality $k \says$ indexed by a principal (i.e.\ user or process) $k$.
These \say-based access control logics are designed for systems in which different principals can issue statements that become part of the access control policy.
$k \says \varphi$ is usually rendered as ``$k$ supports $\varphi$'', which can be interpreted to mean that $k$ has issued statements that -- together with some additional information present in the system -- imply $\varphi$.
Different access control logics vary in their account of which additional information may be assumed in deriving the statements that $k$ supports.

Van Hertum et al.\ \cite{ijcai16Cramer} have recently proposed a multi-agent variant of autoepistemic logic, called \emph{Distributed Autoepistemic Logic with Inductive Definitions} (\DACL), to be used as a \say-based access control logic. Autoepistemic logic is a non-monotonic logic originally designed for reasoning about knowledge bases and motivated by the principle that an agent's knowledge base completely characterizes what the agent knows \cite{mo85}. By applying the semantic principles of autoepistemic logic to characterize the \say-modality, \DACL allows us to derive a statement of the form $\neg k \says \phi$ on the basis of the observation that $k$ has not issued statements implying $\phi$. As explained in Section \ref{sec:motivation}, supporting reasoning about such negated \say-statements allows \DACL to model access denials straightforwardly.

Van Hertum et al.\ have extended multiple semantics of autoepistemic logic to \DACL, but have argued that the well-founded semantics is to be prefered in the application of \DACL to access control. In this paper we therefore restrict ourselves to the well-founded semantics of \DACL.

When applying \DACL to access control, the access control policy consists of a separate set of \DACL formulas for each principal in the system, where the set of formulas of each principal consists of the statements issued by that principal. A principal $k$ has access right to a resource $r$ if and only if the owner $j$ of that resource supports the formula $\textit{access}(k,r)$, i.e.\ iff the \DACL formula $j \says \textit{access}(k,r)$ is true in the well-founded model of the access control policy.

We define a query-driven decision procedure for \DACL, which -- under the assumption of a finite domain -- allows  one to determine the truth value of a formula in the well-founded model of a \DACL access control policy, i.e.\ to determine access rights. This decision procedure is designed in such a way that it avoids redundant information flow between principals, which ensures that the need-to-know principle of computer security \cite{Sandhu94} is not violated, and which additionally reduces privacy concerns. This decision procedure is implemented with the help of the \idp system \cite{WarrenBook/DeCatBBD14}, a knowledge base system for the language of first-order logic with inductive definitions.

The decision procedure that we define has in the worst case an exponential runtime. This means that it is not practicable to build an access control system that implements this decision procedure without including heuristics to optimize response time and a principled approach for dealing with situations when access cannot be determined within a reasonable amount of time (see Section VII of Cramer et al.\ \cite{NP-complete} for an example of such an approach in a somewhat different access-control setting). For this reason, we regard the contribution of this paper to be mainly conceptual: The defined decision procedure is a proof of concept that increases our understanding of dAEL(ID) by providing an algorithmic characterization of the well-founded semantics of dAEL(ID). This algorithmic characterization complements in a conceptually fruitful way the semantic definition from Van Hertum et al.\ \cite{ijcai16Cramer} which is based on a fixpoint construction on abstract structures.

The rest of the paper is organized as follows. 
In Section \ref{sec:dacl}, we define \DACL and motivate its application to access control.
In Section \ref{sec:IDP}, we introduce the \idp system and its language FO(ID). 
In Section \ref{sec:query}, we present a query mechanism for determining access rights while avoiding redundant information flow between principals. Section \ref{sec:related} discusses related work. Section \ref{sec:future} concludes the paper and presents possible future work.

\section{Distributed Autoepistemic Logic with Inductive Definitions}
\label{sec:dacl}

Van Hertum et al.\ \cite{ijcai16Cramer} have used two notational variants of \DACL: In the first one, the modality of the logic is written as $K_A \phi$, following the standard notation in autoepistemic logic. In the second one, it is written as $A \says \phi$, following the standard notation in access control logic. In this paper, we only use the notation $A \says \phi$.







\subsection{\DACL Syntax}
We assume that a set $\A$ of principals and a first-order vocabulary $\Sigma$ consisting of function and predicate symbol with fixed arity is fixed throughout this paper. As usual, $0$-ary function symbols play the role of constants, and $0$-ary predicate symbols play the role of propositional variables. Terms are built from function symbols and variables in the usual manner.

\begin{definition}  \DACL formulas are defined by the following EBNF rule, where $P$ denotes a predicate symbol, $t$ a term and $x$ a variable:
 $$\varphi ::= P(t,\dots,t) \mid t = t \mid \neg \varphi \mid (\varphi \land \varphi) \mid \forall x \; \varphi \mid t ~\says~ \varphi $$
\end{definition}

The symbols $\lor$, $\Rightarrow$, $\Leftrightarrow$ and $\exists$ are treated as abbreviations in the standard way. We follow the standard conventions for dropping brackets when this does not cause unclarity.

The intuitive reading of $t \says \phi$ is ``$t$ is a principal and $t$ supports $\phi$''. So if the term $t$ does not denote a principal, $t \says \phi$ will be interpreted to be false.

\begin{definition} A \emph{\say-atom} or \emph{modal atom} is a formula of the form $t \says \varphi$. A \emph{\say-literal} is a \say-atom $t \says \varphi$ or its negation $\neg t \says \varphi$.
\end{definition}

As motivated in Section \ref{sec:motivation} below, \DACL contains a construct for inductive definitions:

\begin{definition}
We define a \DACL \emph{inductive definition} $\Delta$ to be a finite set of rules of the form $\forall \overline{x}:P(\overline{x}) \leftarrow \varphi(\overline{y})$, where $\overline{y} \subset \overline{x}$ and $\varphi(\overline{y})$ is a \DACL formula.  $P(\overline{x})$ is called the \emph{head} and $\varphi(\overline{y})$ the \emph{body} of the rule $\forall \overline{x}:P(\overline{x}) \leftarrow \varphi(\overline{y})$.
\end{definition}

\begin{definition}
A \DACL \emph{theory} $T$ is a set that consists of \DACL formulas and \DACL inductive definitions.
\end{definition}

In a distributed setting, different principals can issue statements that become part of the access control policy. A \DACL theory as defined above only represent the set of statements of the access control policy issued by a single principal. In order to represent the full access control policy, we use the notion of a \emph{distributed theory}:

\begin{definition}
 A \emph{distributed theory} $\T$ is an indexed family $(\T_A)_{A \in \A}$, where each $\T_A$ is a \DACL theory.
\end{definition}

\subsection{Semantics}
\label{sec:semantics}

Van Hertum et al.\ \cite{ijcai16Cramer} have defined various semantics for \DACL using Approximation Fixpoint Theory \cite{aaai/DMT98}, but have argued for the use of the well-founded semantics in the application of \DACL to access control. In this paper, we define a decision procedure for \DACL with respect to the well-founded semantics, so we only define this semantics. The definition of the semantics involves a lot of technical machinery, but for a reader familiar with autoepistemic logic, it is enough to know that the well-founded semantics of \DACL is an extension of the well-founded semantics of autoepistemic logic \cite{DeneckerMT03} to the multi-agent case under the assumption of mutual positive and negative introspection between the agents. We motivate this mutual introspection below in Section \ref{sec:motivation}. Note that the well-founded semantics of \DACL is defined over a fixed domain $D$, which can be either finite or infinite (but for the decision procedure in Section \ref{sec:query}, $D$ is assumed to be finite).

For defining the well-founded semantics of \DACL, we use the methodology of \emph{Approximation Fixpoint Theory} that Denecker et al.\ \cite{aaai/DMT98} used to define the well-founded semantics of autoepistemic logic.
This methodology is based on the idea of approximating the knowledge of an agent using a three-valued valuation, in which formulas, inductive definitions and theories may not only be true or false but also undefined. The logical connectives combine these three truth values based on Kleene's truth tables \cite{Kleene38}. 

We use truth values $\ltrue$ for truth, $\lfalse$ for falsity, and $\lunkn$ for undefined.
The truth order $<_t$ on truth values is induced
by $\lfalse<_t\lunkn<_t\ltrue$. The precision order $<_p$ on truth values is
induced by $\lunkn<_p\ltrue, \lunkn<_p\lfalse$.
We define $\ltrue^{-1}=\lfalse, \lfalse^{-1}=\ltrue$ and $\lunkn^{-1}=\lunkn$.

A \emph{structure} is defined as usual in first-order logic:

\begin{definition}
A \emph{structure} $\I$ consists of a set $D$, called the \emph{domain} of $\I$, an assignment that maps every $n$-ary predicate symbol of $\Sigma$ to a subset of $D^n$ and an assignment that maps every $n$-ary function symbol of $\Sigma$ to a function $D^n \rightarrow D$.
\end{definition}

A structure formally represents a potential state of affairs of the world. The interpretation of a term in a structure is defined as usual.

We assume a domain $D$, shared by all structures, to be fixed throughout the paper. Furthermore, we assume $D$ to contain the set $\A$ of principals. 

The semantics of \DACL is based on the truth assignment of S5 modal logic, extended to the multi-agent case in such a way that mutual positive and negative introspection between agents is satisfied. While S5 modal logic is often used for formalizing the knowledge modality, we make indirect usage of it for formalizing the \say-modality. But for convenience, we will sometimes use knowledge terminology when informally explaining the formal definitions needed for defining \DACL semantics.

The following notion is used to model a single agent's knowledge:

\begin{definition}
A \emph{possible world structure} $\pws$ is a set of structures.
\end{definition}

Note that a possible world structure can be seen as a Kripke structure with the total accessibility relation. It contains all structures that are consistent with an agent's knowledge.

Possible world structures are ordered with respect to the amount of knowledge they contain.
In this sense, possible world structures that contain less structures possess more knowledge:

\begin{definition}
Given two possible world structures $\pws_1$ and $\pws_2$, we define $\pws_1 \leqpws \pws_2$ to hold if and only if $\pws_1 \supseteq \pws_2$.
\end{definition}

In order to model the interaction of the knowledge of multiple agents, we extend the notion of a possible world structure to the multi-agent case as follows:

\begin{definition}
A \emph{distributed possible world structure} (\emph{DPWS}) $\upws = (\upws_A)_{A\in\A}$ is a family consisting of a possible world structure $\upws_A$ for each principal $A \in \A$.
\end{definition}

The knowledge order on possible world structures can be extended pointwise to DPWS's. One DPWS contains more knowledge than another if each principal has more knowledge:
\begin{definition}\label{def:latticeupws}
Given two DPWS's $\upws^1$ and $\upws^2$, we define $\upws^1 \lequpws \upws^2$ iff $\upws^1_A\leqpws \upws^2_A$ for each $A\in \A$.
\end{definition}

 \begin{definition}
 We inductively define a two-valued valuation of \DACL formulas with respect to a DPWS \upws and a structure $\I$ as follows:
 \begin{align*}
 &(P(\bar{t}))^{\upws,\I}=\Tr &&\text{iff~~} \bar{t}^I\in P^I\\
 &(t_1=t_2)^{\upws,\I}=\Tr &&\text{iff~~} t_1^\I = t_2^\I\\
 &(\varphi_1 \land \varphi_2)^{\upws,\I}=\Tr &&\text{iff~~} (\varphi_1)^{\upws,\I}=\Tr\text{ and } (\varphi_2)^{\upws,\I}=\Tr\\
 &(\neg \varphi)^{\upws,\I}=\Tr &&\text{iff~~} (\varphi)^{\upws,\I}=\Fa\\
 &(\forall x \; \varphi)^{\upws,\I}=\Tr &&\text{iff for each } d \in D, \; (\varphi[x/d])^{\upws,\I}=\Tr\\
 &(t \says \varphi)^{\upws,\I}=\Tr &&\text{iff~~ $t^I \in \A$ and $\varphi^{(\upws,J)}=\Tr$ for all } J\in \upws_{t^I}
 \end{align*}
 \end{definition}

 Inductive definitions generally define only some of the predicates of a language, while the remaining predicates of the language function as parameters:

 \begin{definition}
 Let $\Delta = \{P_1(\bar t_1) \leftarrow \phi_1, \dots, P_n(\bar t_n) \leftarrow \phi_n\}$ be an inductive definition. Then $\defs(\Delta)$ is defined to be $\{P_1,\dots,P_n\}$ and is called the set of \emph{defined predicates} of $\Delta$. The set of predicates in $\Sigma$ that are not in $\defs(\Delta)$ is denoted $\pars(\Delta)$ and is called the set of \emph{parameters} of $\Delta$.
 \end{definition}

In order to approximate the agents' knowledge in a three-valued setting, we use \emph{distributed belief pairs} that consist of a conservative bound $\ubp^c$ and a liberal bound $\ubp^l$ of each agent's knowledge, i.e.\ it specifies what each agent knows for certain and what each agent possibly knows:

\begin{definition}
A \emph{distributed belief pair} \ubp is a pair $(\ubp^c,\ubp^l)$ of two DPWS's $\ubp^c$ and $\ubp^l$ such that $\ubp^c \leq_K \ubp^l$.
\end{definition}

The knowledge order $\leq_K$ on DPWS's induces a precision order $\leq_p$ on distributed belief pairs:

\begin{definition}
Given two distributed belief pairs $\ubp_1$ and $\ubp_2$, we define $\ubp_1 \leq_p \ubp_2$ to hold iff $\ubp_1^c \leq \ubp_2^c$ and $\ubp_2^l \leq \ubp_1^l$.
\end{definition}

Intuitively, $\ubp_1 \leq_p \ubp_2$ means that $\ubp_2$ characterizes the knowledge of the principals more precisely than $\ubp_1$.

\begin{definition}
\label{def:3valform}
We inductively define a three-valued valuation of \DACL formulas with respect to a distributed belief pair \ubp and a structure $\I$ as follows:
\begin{align*}
(P(\ttt))^{\ubp,I} & = &&  \left\{\begin{array}{ll}
                                 \Tr & \text{~if~ } \bar{t}^I\in P^I\\
                                 \Fa & \text{~if~ } \bar{t}^I\not \in P^I
                                \end{array}\right.\\
(\neg \varphi)^{\ubp,I}& = &&(\varphi^{\ubp,I})^{-1}\\
(\varphi\land \psi)^{\ubp,I}& = &&\glb_{{\leq_t}} (\varphi^{\ubp,I},\psi^{\ubp,I})\\
(\forall x \; \varphi)^{\ubp,I}& = &&\glb_{{\leq_t}} \{\varphi[x/d]^{\ubp,I}\mid d \in D\}\\
(t \says \varphi)^{\ubp,I}& = &&\left\{\begin{array}{ll}
                                 \Tr & \text{if ~} t^I \in \A \text{ and } \\
                                 &\hspace{3.2mm} \varphi^{\ubp,I'}=\Tr \text{ for all $I'\in \ubp_{t^I}^c$}\\
                                 \Fa & \text{if ~} t^I \notin \A \text{ or } \\
                                 &\hspace{3.2mm} \varphi^{\ubp,I'}=\Fa \text{ for some $I'\in \ubp_{t^I}^l$}\\
                                 \Un & \text{~otherwise}
                                \end{array}\right.
\end{align*}
\end{definition}

As explained in Section \ref{sec:motivation} below, inductive definitions in \DACL are interpreted according to the well-founded semantics for inductive definitions, as defined for example in \cite{KR/DeneckerV14}. The well-founded model of an inductive definition $\Delta$ is always defined relative to a context $\mc{O}$, which is an interpretation of the predicate symbols in $\pars(\Delta)$. We denote the well-founded model of $\Delta$ relative to $\mc{O}$ by $\wfm_\Delta(\mc{O})$.

Inductive definitions in \DACL may contain the \say-modality in the body. Since the definition of the well-founded model in \cite{KR/DeneckerV14} is only defined for inductive definition over a first-order language without any modality, we need to say something about how to interpret the \say-modality in the body. Just like formulas, we evaluate inductive definitions with respect to a DPWS \upws and a structure $\I$. The DPWS \upws assigns a truth-value to every formula of the form $k \says \phi$. When evaluating an inductive definition $\Delta$ with respect to \upws and $\I$, it should get evaluated in the same way as the inductive definition $\Delta^\upws$, which is defined to be $\Delta$ with all instances of formulas of the form $k \says \phi$ replaced by \Tr or \Fa according to their interpretation in $\upws$.

This motivates the following definition of a three-valued valuation of \DACL inductive definitions with respect to a DPWS \upws and a structure $\I$:


\begin{definition}
\label{def:3valID}
We define a three-valued valuation of \DACL inductive definitions with respect to a distributed belief pair \ubp and a structure $\I$ as follows:
\begin{align*}
&\Delta^{\ubp,I}= &&\left\{\begin{array}{ll}
                                 \Tr & \text{ if } I=\wfm_{\Delta^{\ubp}}(I|_{Par(\Delta)}) \\
                                 \Fa & \text{ if } I\not\geq_p \wfm_{\Delta^{\ubp}}(I|_{Par(\Delta)}) \\
                                 \Un & \text{otherwise}
                                \end{array}\right.
\end{align*}
where
$\Delta^\ubp$ is the definition $\Delta$ with all formulas $t \says \varphi$ replaced by $\Tr$, $\Fa$ or $\Un$, according to their interpretation in $\ubp$.
\end{definition}

To understand this three-valued valuation of \DACL inductive definitions informally, remark that in a partial context ($\ubp$ is three-valued), we cannot yet evaluate the exact value of the defined predicates in the definition. We can, however, using a three-valued valuation of the definition, obtain an approximation $\wfm_{\Delta^{\ubp}}(I|_{Par(\Delta)})$  of their value.
We return $\Tr$ if this approximation is actually two-valued and equal to $I$, $\lunkn$ if $I$ is still consistent with (but not equal to) this approximation and $\lfalse$ otherwise.

We can combine the three-valued valuations for formulas and inductive definitions into a three-valued valuation of a single agent's theory as follows:

%
%

\begin{definition}
\label{def:3valtheo}
We define a three-valued valuation of \DACL theories with respect to a distributed belief pair \ubp and a structure $\I$ as follows:
$$T^{\ubp,I} := \glb_{\leq_t} (\{\phi^{\ubp,I}|\phi \in T\} \cup \{\Delta^{\ubp,I}|\Delta \in T\})$$
\end{definition}

Using this three-valued valuation of \DACL theories, we can define an operator $\ubprevision$ on distributed belief pairs:

\begin{definition}
\label{def:D*}
$\ubprevision(\ubp):=(\ubprevisionL(\ubp),\ubprevisionU(\ubp))$, where
\begin{align*}
&\ubprevisionL(\ubp):=(\{I\mid   (\T_A)^{\ubp,I} \neq \Fa \})_{A\in \A} \\
&\ubprevisionU(\ubp):=(\{I\mid   (\T_A)^{\ubp,I} = \Tr \})_{A\in \A}
\end{align*}
\end{definition}

In order to formally define the well-founded model, we first need to define the \emph{stable operator} $S_\T$ that maps a DPWS to a DPWS:

\begin{definition}
$S_\T(\upws)$ is defined to be the least fixpoint of the operator $O$ that maps a DPWS $\upws'$ to the DPWS $O(\upws):=\ubprevision(\upws',\upws)_1$, i.e. to the first element of the distributed belief pair $\ubprevision(\upws',\upws)$.
\end{definition}

Now we are ready to define the well-founded model of a distributed theory, the central notion of \DACL semantics:

\begin{definition}
Let $\T$ be a distributed theory. The well-founded model of $\T$, denoted $\wfm(\T)$, is the least precise (i.e.\ $\leq_p$-minimal) distributed belief pair \ubp such that $S_\T(\ubp^c) = \ubp^l$ and $S_\T(\ubp^l) = \ubp^c$.
\end{definition}

We say that a distributed theory logically implies a formula $\phi$ iff $\phi^{\wfm(\T),I} = \Tr$ for every structure $I$.

Note that for a formula $\phi$ of the form $k \says \psi$ or $\neg k \says \psi$, the value of $\phi^{\wfm(\T),I}$ does not depend on $I$. We therefore sometimes write $\phi^{\wfm(\T)}$ instead of $\phi^{\wfm(\T),I}$ for such $\phi$.

\subsection{Motivation for \DACL}
\label{sec:motivation}
Van Hertum et al.\ \cite{ijcai16Cramer} have motivated the applicability of \DACL as an access control logic by discussing possible use cases, i.e. by illustrating how \DACL can be applied in certain access control scenarios. In this section we add to this motivation by use cases a more principled motivation that clarifies the advantages of \DACL over other \say-based access control logics.

An \emph{access control policy} is a set of norms defining which principal is to be granted access to which resource under which circumstances.
Specialized logics called \emph{access control logics} were developed for representing policies and access requests and reasoning about them.
A general principle adopted by most logic-based approaches to access control is that access is granted iff it is logically entailed by the policy.

There is a large variety of access control logics, but most of them use a modality $k \says$ indexed by a principal $k$ \cite{Genovese12}.
\say-based access control logics are designed for systems in which different principals can issue statements that become part of the access control policy.
$k \says \varphi$ is usually explained informally to mean that $k$ supports $\varphi$ \cite{Abadi08,Garg12,Genovese12}.
This means that $k$ has issued statements that -- together with additional information present in the system -- imply $\varphi$.
Different access control logics vary in their account of which rules of inference and which additional information may be used in deriving  statements that $k$ supports from the statements that $k$ has explicitly issued.


Many state-of-the-art \say-based access control logics, e.g.\ Garg's BL \cite{Garg12}, do not provide the means for deriving statements of the form $\neg k \says \varphi$ or $j \says (\neg k \says \varphi)$.
However, being able to derive statements of the form $\neg k \says \varphi$ and $j \says (\neg k \says \varphi)$ makes it possible to model access denials naturally in a \say-based access control logic: Suppose $A$ is a professor with control over a resource $r$, $B$ is a PhD student of $A$ who needs access to $r$, and $C$ is a postdoc of $A$ supervising $B$. $A$ wants to grant $B$ access to $r$, but wants to grant $C$ the right to deny $B$'s access to $r$, for example in case $B$ misuses her rights.
A natural way for $A$ to do this  using the \say-modality is to issue the statement $(\neg C \says \neg \access(B,r)) \Rightarrow \access(B,r)$. This should have the effect that $B$ has access to $r$ unless $C$ denies him access.
However, this effect can only be achieved if our logic allows $A$ to derive $\neg C \says \neg \access(B,r)$ from the fact that $C$ has not issued any statements implying $\neg \access(B,r)$.

The derivation of $\neg C \says \neg \access(B,r)$ from the fact that $C$ has not issued any statements implying $\neg \access(B,r)$ is non-monotonic: If $C$ issues a statement implying $\neg \access(B,r)$, the formula $\neg C \says \neg \access(B,r)$ can no longer be derived. In other words, adding a formula to the access control policy causes that something previously implied by the policy is no longer implied. Existing \say-based access control logics are monotonic, so they cannot support the reasoning described above for modelling denial with the \say-modality.

In order to derive statements of the form $\neg k \says \varphi$, we have to assume the statements issued by a principal to be a complete characterization of what the principal supports. This is similar to the motivation behind Moore's autoepistemic logic (AEL) to consider an agent's theory to be a complete characterization of what the agent knows \cite{mo85,ai/Levesque90,ijcai/Niemela91,nonmon30/DeneckerMT11}.
This motivates an application of AEL to access control.

However, AEL cannot model more than one agent. In order to extend it to the multi-agent case, one needs to specify how the knowledge of the agents interacts. Most state-of-the-art access control logics allow $j \says (k \says \varphi)$ to be derived from $k \says \varphi$, as this is required for standard delegation to be naturally modelled using the \say-modality. In the knowledge terminology of AEL, this can be called mutual positive introspection between agents. In order to also model denial as described above, we also need mutual negative introspection, i.e. that $j \says (\neg k \says \varphi)$ to be derived from $\neg k \says \varphi$. Van Hertum et al.\ \cite{ijcai16Cramer} have defined the semantics of \DACL in such a way that mutual positive and negative introspection between the agents is ensured.

\DACL also incorporates inductive definitions, thus allowing principals to define access rights and other properties relevant for access control in an inductive way. Inductive (recursive) definitions are a common concept in all branches of mathematics. Inductive definitions in \DACL are intended to be understood in the same way as in the general purpose specification language \fodot of the \idp system \cite{WarrenBook/DeCatBBD14}. Denecker \cite{Denecker:CL2000} showed that in classical logics, adding definitions leads to a strictly more expressive language.

Because of their rule-based nature, formal inductive definitions also bear strong similarities in syntax and formal semantics with logic programs. A formal inductive definition could also be understood intuitively as a logic program which has arbitrary formulas in the body and which defines only a subset of the
predicates in terms of parameter predicates not defined in the definition.

Most of the semantics that have been proposed for logic programs can be adapted to inductive definitions. Denecker and Venneckens \cite{KR/DeneckerV14} have argued that the well-founded semantics correctly formalizes our intuitive understanding of inductive definitions,
and hence that it is actually the \emph{right} semantics. Following them, we use the well-founded semantics for inductive definitions.
%



\section{\foid and the \idp-system}\label{sec:IDP}
The decision procedure defined in the next section is based on the \idp system, so we briefly describe this system and its language \foid.

\subsection{Why \idp?}

\idp \cite{WarrenBook/DeCatBBD14} is a Knowledge Base System which combines a declarative specification (\emph{knowledge base}), written in an extension of first-order logic, with an imperative management of the specification via the Lua \cite{SPE/IerusalimschyFC96} scripting language. The extension of first-order logic supported by \idp allows for inductive definitions. As explained and motivated in Section \ref{sec:motivation}, \DACL also supports inductive definitions. This makes the usage of \idp as a basis for the decision procedure a natural choice. 

\idp supports multiple \emph{inferences} that can be used to perform a range of reasoning tasks on a given specification. We make use of two of \idp's inferences, defined in Section \ref{sec:foid} below, in order to perform the meta-reasoning about a principal's \DACL theory that is necessary to determine which queries to other principals are really necessary in order to resolve a query asked to the principal.

%


\subsection{\foid and some \idp inferences}
\label{sec:foid}
The specification language supported by \idp is an extension of \fo with types, inductive definitions, aggregates, arithmetic and partial functions, denoted FO(T,ID,Agg,Arit,PF) \cite{KBParadigm16}. We only make use of the subset of FO(T,ID,Agg,Arit,PF) called \foid, which extends \fo only with inductive definitions. The formal definition of \foid syntax is the standard definition of \fo syntax extended by the following definition of inductive definitions: An \emph{inductive definition} $\Delta$ is a set of rules of the form $\forall \overline{x}:P(\overline{x}) \leftarrow \varphi(\overline{y})$, where $\overline{y} \subset \overline{x}$ and $\varphi(\overline{y})$ is an \fo formula. Just as in \DACL, inductive definitions are given the well-founded semantics of inductive definitions \cite{KR/DeneckerV14}. An \foid theory is a set of inductive definitions and \fo formulas.

The \idp inferences for \foid that we make use of are defined for \emph{finite partial structures}.
Before we define formally what a partial structure is, we define the concept of a \emph{partial set}, a generalization of a set in a three-valued context:

\begin{definition}
 A \emph{partial set} on the domain $D$ is a function from $D$ to $\{\Tr,\Fa,\Un\}$, where $\Tr$, $\Fa$ and $\Un$ stand for the three truth-values \emph{true}, \emph{false} and \emph{undefined}.
\end{definition}

A partial set is two-valued (or total) if $\Un$ does not belong to its range.

Given a vocabulary $\Sigma$, a partial structure gives an interpretation to the elements of~$\Sigma$:

\begin{definition}
 A \emph{partial structure} over our fixed vocabulary $\Sigma$ is a tuple $(D, \mathcal{I})$, where the domain $D$ is a set, and $\mathcal{I}$ is an assignment function that assigns an interpretation to each symbol in $\Sigma$. For a predicate symbol $P$ of arity $n$, the interpretation $P^{\mathcal{I}}$ is a partial set on the domain $D^n$;
for a function symbol $f$ of arity $n$,  $f^\mathcal{I}$ is a function from $D^n$ to $D$.
\end{definition}

When the predicate symbol $P$ has arity $0$, i.e.\ is a propositional variable, $P^{\mathcal{I}}$ is just an element of $\{\Tr,\Fa,\Un\}$.

We call a partial structure $S=(D,\mathcal{I})$ \emph{finite} iff its domain $D$ is finite. We call a partial structure \emph{total} iff $P^\mathcal{I}$ is total for all $P \in \Sigma$.

The interpretation of terms $t^{\mathcal{I}}$ and the satisfaction relation $\models$ for total structures $S \models \varphi$ are defined as usual.

A precision order can be defined on partial structures:
\begin{definition}
 Given two partial structures $S = (D,\mathcal{I})$ and $S'=(D,\mathcal{I}')$, we write $S \leq_p S'$ (and say $S'$ \emph{is more precise than} $S$, or $S'$ \emph{expands} $S$) iff for every function symbol $f$, $f^{\mathcal{I}'} = f^\mathcal{I}$, and for every predicate symbol $P$ of arity $n$ and every tuple $\bar d \in D^n$ of domain elements such that $P^\mathcal{I}(\bar d) \neq \Un$, we have $P^{\mathcal{I}'}(\bar d) = P^\mathcal{I}(\bar d)$.
\end{definition}

We now define the two \idp inferences that we make use of. The first one, which is called \texttt{sat} in the \idp system, determines whether a given finite partial structure is a partial model of a given theory:

\begin{definition} Let $S$ be a partial structure and $\mathcal{T}$ an \foid theory. We say $S$ is a \emph{partial model} for $\mathcal{T}$ if and only if there exists a total structure $S^\prime \geq_p S$ such that $S^\prime \models \mathcal{T}$.
\end{definition}





The second \idp inference that we make use of, which is called \texttt{unsatstructure} in the \idp system, picks a minimal partial structure inconsistent with a given theory and less precise than a given finite partial structure:

\begin{definition} Let $S$ be a partial structure and $\mathcal{T}$ be an \foid theory. We define $min\_incons\_set(\mathcal{T},S)$ to be the set of $\leq_p$-minimal partial structures $S^\prime \leq_p S$ such that $S^\prime$ is not a partial model of $\mathcal{T}$. 
\end{definition}

If the input structure $S$ is not a partial model of the input theory $\mathcal{T}$, then \linebreak $min\_incons\_set(\mathcal{T},S)$ is always non-empty, and \texttt{unsatstructure} picks an element from it and returns it. If $S$ is a partial model of $\mathcal{T}$, \texttt{unsatstructure} throws an error.

\section{Decision Procedure}\label{sec:query}

In this section, we define a query-driven decision procedure for \DACL, which allows to determine access rights while avoiding redundant information flow between principals in order to enhance security and reduce privacy concerns. This decision procedure is implemented with the help of the \idp system. Given that \idp can only work with finite domains, the decision procedure also assumes the domain $D$ to be finite.\footnote{Given that propositional logic has the same expressive power as first-order logic over a finite domain, the decision procedure could in theory also be viewed as a decision procedure for the propositional fragment of \DACL. But since first-order logic over a finite domain can model the same scenarios more concisely and more naturally than propositional logic, we stick to the first-order variant of \DACL with a finite-domain assumption.} For simplicity, we assume that for every principal there is a constant symbol referring to that principal, and that the $t$ in every formula of the form $t \says \phi$ is such a  constant symbol. This simplification could be removed, but would make the description of the decision procedure much more complicated.

The decision procedure is query-driven in the following sense: A query in the form of a \DACL formula $\phi$ is posed to a principal $A$. $A$ determines whether her theory contains enough information in order to verify $\phi$. It can happen that $A$ cannot verify $\phi$ just on the basis of her theory, but can determine that if a certain other principal supports a certain formula, her theory implies the query. For example, $A$'s theory may contain the formula $B \says p \Rightarrow \phi$. In this case, $A$ can forward a remote sub-query to $B$ concerning the status of $p$ in $B$'s theory. If $B$ verifies the sub-query $p$ and informs $A$ about this, $A$ can complete her verification of the original query $\phi$.


\subsection{Motivation for avoiding redundant information flow}
\label{sec:minimization_motivation}
One reason to avoid redundant information flow is to reduce communication overhead. The rest of this section considers an additional motivation for avoiding redundant information flow.

Consider the following distributed theory of the two principals $A$ and $B$:

\[ T_A = \left\{
\begin{array}{l}
     r \land B \says s \Rightarrow p  \\
     r
\end{array}
 \right\} \]
\[ T_B = \left\{
\begin{array}{l}
     s  \\
     \neg s \land A \says p \Rightarrow p
\end{array}
 \right\} \]
\vspace{1.2mm}

In both theories we have a guard, namely, $r$ for theory $T_A$ and $\neg s$ for theory $T_B$. The guards can be checked locally before performing a remote query to other theories. If $A$ is queried about $p$, we can continue with the evaluation and query $T_B$ about the truth value of $s$, since the guard $r$ is true. If $B$ is queried about $p$, on the other hand, we do not need to perform any remote query since it will always fail due to the guard being false in the theory.

If $B$ nevertheless were to send the remote subquery $p$ to $A$, this would be an unnecessary sub-query. Since $B$ does not actually need to know whether $A$ supports $p$, this would violate the need-to-know principle \cite{Sandhu94}, which states that a principal should only be given those accesses and be provided with that non-public information which the principal requires to carry out her responsibilities. 
Additionally, it is reasonable to assume that for privacy considerations, the principals do not want to disclose their full access control policies to other principals, but only the parts that are required to verify a given access request. So there are both security and privacy reasons for $B$ not to send the remote subquery $p$ to $A$.


In general, more complex behaviors rather than guards can occur in a distributed theory. The decision procedure we define avoids redundant communication even when more complex reasoning is required to determine which sub-queries have a chance of leading to a verification of the primary query and which subqueries are certainly not useful. As discussed in Section \ref{sec:complexity}, this ideal avoidance of redundant communication is computationally very expensive, so in a practically applicable system, a trade-off between the security and privacy motivation for avoiding redundant communication on the one hand and computational cost on the other hand would need to be found. Nevertheless, we consider our ideal avoidance of redundant communication an interesting proof of concept as a foundation for further research.

The decision procedure that determines whether a query $\alpha$ is true given a distributed theory $\T$ is composed of two distinct modules. The first module, the \emph{Query Minimization Procedure}, looks at the theory of the agent to whom the query is directed, and determines minimal sets of remote calls to other theories that could verify the query. The second module, the \emph{Communication Procedure}, takes care of communication between the principals, including the handling of the loops that may occur.

\subsection{Query Minimization Procedure}



\subsubsection{Translation Mechanism.}

In order to implement a query mechanism for \DACL in IDP we need to translate \DACL theories to \foid theories.
The only syntactic construct of \DACL that does not exist in \foid is the \say-modality. So when translating a \DACL theory $T$ to an \foid theory $\mathcal{T}$, we need to replace each \say-atoms in $T$ by some first-order formula. For this purpose, we extend the vocabulary $\Sigma$ to an extended vocabulary $\Sigma'$ by adding to it new propositional variables of the form $p^+_{A\_{\verb=says=\_\varphi}}$, $p^-_{A\_{\verb=says=\_\varphi}}$ and $w_{A\_{\verb=says=\_\phi}}$ for every modal statement $A \says \varphi$ of \DACL.


Before we formally define the translation mechanism, let us first motivate why we have the three different propositional variables $p^+_{A\_{\verb=says=\_\varphi}}$, $p^-_{A\_{\verb=says=\_\varphi}}$ and $w_{A\_{\verb=says=\_\phi}}$ for translating different occurrences of the same \say-atom $A \says \phi$. First, note that the well-founded semantics of \DACL evaluates \say-atoms in a three-valued way. The propositional variables $p^+_{A\_{\verb=says=\_\varphi}}$ and $p^-_{A\_{\verb=says=\_\varphi}}$ are used to model the three-valued valuation of $A \says \phi$ in the two-valued logic \foid: On the precision order $<_p$ on the three truth values $\ltrue$ (\emph{true}), $\lfalse$ (\emph{false}) and $\lunkn$ (\emph{undefined})
induced by $\lunkn<_p\ltrue$ and $\lunkn<_p\lfalse$, the propositional variable $p^+_{A\_{\verb=says=\_\varphi}}$ represents the upper bound for the truth value of $A \says \varphi$ and $p^-_{A\_{\verb=says=\_\varphi}}$ the lower bound. For this reason, we replace every positive occurrence of $A \says \phi$ by $p^+_{A\_{\verb=says=\_\varphi}}$, and every negative occurrence by $p^-_{A\_{\verb=says=\_\varphi}}$. Given that occurrences of a formula in an inductive definition cannot be meaningfully termed only positive or only negative, we first replace occurrences of $A \says \phi$ in an inductive definition by $w_{A\_{\verb=says=\_\phi}}$ and add two implications to the theories that express the equivalence between $w_{A\_{\verb=says=\_\psi}}$ and $A \says \phi$.

The translation function $t$ only performs this first step of the translation mechanism:

\begin{definition}Let $T$ be a \DACL theory. We define $t(T)$ to be a \DACL theory equivalent to $T$, constructed as follows:

For every modal atom $A \says \varphi$ occurring in the body of an inductive definition in theory $T$:
\begin{itemize}
  \item  Replace $A \says \varphi$ by the propositional variable $w_{A\_{\verb=says=\_\varphi}}$
  \item  Add to $t(T)$ the two formulae $w_{A\_{\verb=says=\_\varphi}} \Rightarrow A \says \varphi$ and $A \says \varphi \Rightarrow w_{A\_{\verb=says=\_\varphi}}$.
\end{itemize}

\end{definition}

We next introduce the notion of polarity necessary to further translate \DACL theories into \foid theories. 


\begin{definition} Let $\phi$ be a \DACL formula. The \emph{polarity} of an occurrence of a subformula of $\phi$ is defined recursively as follows:
\begin{itemize}
\item The occurrence of $\phi$ in $\phi$ is a \emph{positive} occurrence.
\item Given a positive (resp. \emph{negative}) occurrence of the subformula $\neg \psi$ of $\phi$, the occurence of $\psi$ in this occurrence of $\neg \psi$ is \emph{negative} (resp. \emph{positive}) in $\phi$.
\item Given a \emph{positive} (resp. \emph{negative}) occurrence of the subformula $\psi \land \chi$ of $\phi$, the occurrences of $\psi$ and $\chi$ in this occurrence of $\psi \land \chi$ are both \emph{positive} (resp. \emph{negative}) in $\phi$.
\end{itemize}
\end{definition}

\begin{definition}
 Let $T$ be a \DACL theory, let $\phi \in T$. We call a positive (resp. negative) occurrence of a subformula $\psi$ of $\phi$ a \emph{positive} (resp. \emph{negative}) \emph{occurrence} of $\psi$ in~$T$.
\end{definition}

Now we can define the translation function $\tau$ from \DACL theories to \foid theories:

%

\begin{definition} Let $T$ be a \DACL theory. $\tau(T)$ is constructed from $t(T)$ by performing the following replacements for every \say-atom $A \says \varphi$ occurring in $t(T)$ that is not a subformula of another \say-atom:
\begin{itemize}
  \item Replace every positive occurrence of $A \says \varphi$ in $T$ by $p^+_{A\_{\verb=says=\_\varphi}}$.
  \item Replace every negative occurrence of $A \says \varphi$ in $T$ by $p^-_{A\_{\verb=says=\_\varphi}}$.
\end{itemize}

\end{definition}


We will illustrate the translation procedure with a simple example, which we will use as a running example to be extended throughout the section.

\begin{example}\label{ex:example} Let $\A=\{A,B,C\}$, and let the distributed theory $\T$ consist of the folowing three \DACL theories:

  \[ T_A = \left\{
\begin{array}{ll}
   &  \{ \ p \leftarrow B ~\says~ p,  \\
   & p \leftarrow r \ \} \\
   &  p \land s \land B ~\says~ z \Rightarrow z \\
   & r \lor \neg r \Rightarrow s\\
   & B ~\says~ r \lor \neg (B ~\says~ r) \Rightarrow z
\end{array}
 \right\} \]
\[ T_B = \left\{
\begin{array}{ll}
   &  p  \\
   &  C ~\says~ z \Rightarrow z \\
   & C ~\says~ r \Rightarrow r
\end{array}
 \right\} \]
\[ T_C = \left\{
\begin{array}{ll}
   &  \neg(B ~\says~ z) \Rightarrow z  \\
   &  B ~\says~ r \Rightarrow r
\end{array}
 \right\} \]

We translate these theories as follows:

\[ \tau(T_A) = \left\{
\begin{array}{ll}
& \{ \ p \leftarrow w_{B\_{\texttt{says}\_p}},	\\
& p \leftarrow r \ \} \\
& w_{B\_{\texttt{says}\_p}} \Rightarrow p^+_{B\_{\texttt{says}\_p}}\\
& p^-_{B\_{\texttt{says}\_p}} \Rightarrow w_{B\_{\texttt{says}\_p}}\\
& p \land s \land p^-_{B\_{\texttt{says}\_z}} \Rightarrow z\\
& r \lor \neg r \Rightarrow s\\
& p^-_{B\_{\texttt{says}\_r}} \lor \neg p^+_{B\_{\texttt{says}\_r}} \Rightarrow z\\
\end{array}
 \right\} \]
\[ \tau(T_B) = \left\{
\begin{array}{ll}
& p\\
& p^-_{C\_{\texttt{says}\_z}} \Rightarrow z \\
& p^-_{C\_{\texttt{says}\_r}} \Rightarrow r
\end{array}
 \right\} \]
\[ \tau(T_C) = \left\{
\begin{array}{ll}
& \neg p^+_{B\_{\texttt{says}\_z}} \Rightarrow z\\
& p^-_{B\_{\texttt{says}\_r}} \Rightarrow r
\end{array}
 \right\} \]
\end{example}

%
%
%
%

\subsubsection{Query Minimization Procedure.} \label{subsec:minimization}

The query minimization procedure works as follows: given a theory $T$ and a query $\alpha$, the procedure returns a set $\mathbb{L}$ of sets of modal atoms. The intended meaning of $\mathbb{L}$ is as follows: When all modal atoms in a set $L \in \mathbb{L}$ can be determined to be true, the query $\alpha$ succeeds, and $\mathbb{L}$ is the set of all sets $L$ with this property. This means that if $\mathbb{L} = \{\}$, the query necessarily fails, whereas if $\mathbb{L}=\{\{\}\}$ (contains the empty set), the query necessarily succeeds.

A partial structure $S$ over the extended vocabulary $\Sigma'$ contains information about the truth values of the propositional variables of the form $p^-_{A\_{\verb=says=\_\varphi}}$ and $p^+_{A\_{\verb=says=\_\varphi}}$. Taking into account that $p^-_{A\_{\verb=says=\_\varphi}}$ and $p^+_{A\_{\verb=says=\_\varphi}}$ are used to represent the three-valued valuation of $A \says \phi$, this information can also be represented by a set of \say-literals, which we denote $L^S$:


%
%



\begin{definition} For a partial structure $S=(D,\mathcal{I})$, we define $L^S$ to be 
$$\{A \says \varphi \;|\; (p^-_{A\_{\verb=says=\_\varphi}})^\mathcal{I} = \Tr\} \cup  \{\neg A \says \varphi \;|\; (p^+_{A\_{\verb=says=\_\varphi}})^\mathcal{I} = \Fa\}$$
\end{definition}



We say that a \say-atom $A \says \phi$ \emph{occurs directly} in a \DACL theory, if some occurrence of $A \says \phi$ in $\T$ is not a subformula of another \say-atom. In the Query Minimization Procedure, we need to take into account all possible three-valued valuations of the \say-atoms directly occurring in the input \DACL theory $T$. Such a valuation can be represented by a partial structure that contains information only about propositional variables of the form $p^+_{A\_\texttt{says}\_\varphi}$ and $S \models p^-_{A\_\texttt{says}\_\phi}$, and for which this information is coherent in the sense that the truth values assigned to $p^+_{A\_\texttt{says}\_\varphi}$ and $S \models p^-_{A\_\texttt{says}\_\phi}$ are compatible. This is made formally precise in the following definition of the set $\mathbb{S}_T$ that contains all structures that represent three-valued valuations of \say-atoms directly occurring in $T$:

\begin{definition} Let $T$ be a \DACL theory. We define $\mathbb{S}_T$ to be the set containing every partial structure $S=(D,\mathcal{I})$ over vocabulary $\Sigma'$ satisfying the following properties:
\begin{itemize}
 \item $P^{\mathcal{I}^S} = \mathbf{u}$ for every symbol in $\Sigma'$ that is not of the form $p^+_{A\_\texttt{says}\_\varphi}$ or $p^-_{A\_\texttt{says}\_\phi}$ for some \say-atom $A \says \varphi$ occurring in $\tau(T)$.
 \item For every \say-atom $A \says \phi$, $(p^+_{A\_\texttt{says}\_\varphi})^\mathcal{I} \neq \Tr$.
 \item For every \say-atom $A \says \phi$,  $(p^-_{A\_\texttt{says}\_\varphi})^\mathcal{I} \neq \Fa$.
 \item For no \say-atom $A \says \phi$, $(p^+_{A\_\texttt{says}\_\varphi})^\mathcal{I} = \Fa$ and \newline $(p^-_{A\_\texttt{says}\_\varphi})^\mathcal{I} = \Tr$.
\end{itemize}
\end{definition}

We are now ready to define the Query Minimization Procedure. Its pseudo-code is as follows (Algorithm 1). Please note that lines 4 and 5 are implemented using the \idp inferences \texttt{sat} and \texttt{unsatstructure} that we defined in Section \ref{sec:foid}.

\begin{algorithm}[h!t]
\caption{\textsf{Query Minimization Procedure} \label{alg:query}}
\begin{algorithmic}[1]
\REQUIRE{theory $T$, \DACL query $\alpha$ }
\ENSURE{set $\mathbb{L}$ of sets of modal atoms}
    \STATE{$\mathbb{L}$ := $\emptyset$ }
    \STATE{$\mathcal{T}$ := $\tau(T  \cup \{ \neg \alpha \} )$ }
\FOR{ \textbf{each} $S \in  \mathbb{S}_T$}
      \IF{ $S$ is not a partial model of $\mathcal{T}$}
        \STATE{pick a partial structure $S_{min}$ from $min\_incons\_set(\mathcal{T},S)$}
        \STATE{ $\mathbb{L}$ := $\mathbb{L} \cup \{L^{S_{min}}\}$ }\label{lin:last}
    \ENDIF
\ENDFOR
\RETURN{$\mathbb{L}$}
\end{algorithmic}
\end{algorithm}

The algorithm is to be read as follows. A query $\alpha$ asked to theory $T$ is given as input. First (line 2) we 
translate theory $T$ and the negation of the query $\alpha$ into an augmented \foid theory $\mathcal{T}$. Next we iterate over the structures $S \in \mathbb{S}_T$ (lines 3-6).
Line 4 ensures that we limit ourselves to structures $S \in \mathbb{S}_T$ that are not a partial models of $\mathcal{T}$; note that the information in such a structure $S$ together with the information in 
$T$ entails the query $\alpha$. 
Furthermore, note that for such a structure $S$, $min\_incons\_set(\mathcal{T}, S)$ is non-empty. So next (line 5), we pick a structure $S_{min}$ from $min\_incons\_set(\mathcal{T}, S)$; by definition $S_{min}$ is a minimal structure such that $S_{min} \leq_p S$ and $S_{min}$ is not a partial model of $\mathcal{T}$; this means that $S_{min}$ contains a minimal amount of information from $S$ that together with the information in 
$T$ ensures the query $\alpha$ to be true. So the set $L^{S_{min}}$, which represents the same information as a set of \say-literals, is a minimal set of \say-literals that together with the information in $T$ ensure the query $\alpha$ to be true.\footnote{Lemma \ref{lem:minimization} in Appendix B makes this claim more precise.}
Line 6 adds $L^{S_{min}}$ to the set of sets of \say-literals that we output at the end (line 7), after the iteration over the elements of $\mathbb{S}_T$ is completed.


We continue Example \ref{ex:example} to illustrate the query minimization procedure.

\begin{example} \label{ex:query-min} We apply the Query Minimization Procedure to the theory $T_A$ and the query $z$.
First we translate the theory $T_A$ and the negation of the query into $\mathcal{T} = \tau(T_A \cup \{\neg z\})$, as shown in Example \ref{ex:example} with the addition of the formula $\neg z$, since the query does not contain any \say-atoms.
Then we iterate over the structures $S \in \mathbb{S}_{T_A}$.

Let, for example, $S$ be the element of $\mathbb{S}_{T_A}$ that makes $p^-_{B\_{\texttt{says}\_p}}$ and $p^-_{B\_{\texttt{says}\_r}}$ true and everything else undefined. Then $S$ is a not partial model of $\mathcal{T}$, because $p^-_{B\_{\texttt{says}\_p}}$ is inconsistent with $p^-_{B\_{\texttt{says}\_r}} \lor \neg p^+_{B\_{\texttt{says}\_r}} \Rightarrow z$ and $z$. Now $min\_incons\_set(\mathcal{T},S)$ is the set consisting only of the structure $S'$ that makes $p^-_{B\_{\texttt{says}\_p}}$ true and everyting else undefined. So in line 5, we necessarily pick $S_{min}$ to be this structure $S'$. In line 6 we calculate $L^{S'}$ to be $\{B \says r\}$ and add $\{B \says r\}$ to $\mathbb{L}$.


When we iterate over all structures $S \in \mathbb{S}_{T_A}$, the value of $\mathbb{L}$ finally becomes 
$\{\{B \says r\}, \{B \says p, B \says z\},\{\neg B \says r\}\}$.
\end{example}

%
%
%
%

\subsection{Communication and loop handling}

In this subsection we describe the Communication Procedure, which also takes care of the loop-handling. The Communication Procedure calls the Query Minimization Procedure and thereby constitutes our decision procedure for \DACL.

When a query is asked to a principal, the Query Minimization Procedure determines minimal sets of \say-literals that need to be satisfied in order to verify the query. The Communication Procedure then produces remote sub-queries to other principals that can determine the status of the \say-literals.

The Communication Procedure works by dynamically producing a \emph{query graph} and attaching three-valued truth values to the query vertices in it:

\begin{definition}
\label{def:query graph}
A \emph{query graph} is a labelled directed graph with two kinds of vertices and two kinds of edges:
 \begin{itemize}
  \item The first kind of vertices are the \emph{query vertices}. Each query vertex is labelled by a directed query of the form $\langle k:\phi \rangle$, where $k$ is the principal whose theory is being queried and $\phi$ is the formula representing the query. Additionally, a query vertex is potentially labelled by a truth value in $\{\Tr,\Fa,\Un\}$, which represents the currently active valuation of the query at any moment during the execution of the decision procedure.
  \item The second kind of vertices are the \emph{\say-literal set vertices}. Each \say-literal set vertex is labelled by a set of \say-literals, i.e.\ formulas of the form $k \says \phi$ or $\neg k \says \phi$.
  \item The first kind of edges are unlabelled edges going from a query vertex to a \say-literal set vertex. The intended meaning of such an unlabelled edge from $\langle k:\Q \rangle$ to the \say-literal set $L$ is that one way of making $\Q$ true in $k$'s theory is to make all \say-literals in $L$ true.
  \item The second kind of edged are edges labelled by $\Tr$ or $\Fa$, going from a \say-literal set vertex to a query vertex. The intended meaning of such an edge labelled by $\Tr$ or $\Fa$ and going from the \say-literal set $L$ to the query $\langle k:\Q \rangle$ is that $L$ contains the literal $ k \says \Q $ or the literal $\neg k \says \Q$ respectively.
 \end{itemize}
\end{definition}

The query graphs are actually always trees, with the query vertex corresponding to the original query as their root. 

The Communication Procedure starts with a query graph consisting just of the query vertex $\langle A:\alpha \rangle$, where $A$ is the principal to whom the primary query $\alpha$ is asked. Next the Communication Procedure calls the Query Minimization Procedure to add sub-queries to the query graph and attach truth values to them. This procedure is iteratively continued until a truth-value has been attached to the root vertex $\langle A:\alpha \rangle$.

The Communication Procedure is defined via an initialization procedure defined under Algorithm \ref{alg:init}, which calls the main recursive procedure defined under Algorithm~\ref{alg:comm}.

\begin{algorithm}[h!t]
\caption{\textsf{Communication Procedure Initialization} \label{alg:init}}
\begin{algorithmic}[1]
\REQUIRE{distributed theory $\T$, principal $A$, \DACL formula $\Q$}
\ENSURE{truth-value $V \in \{\Tr,\Fa,\Un\}$}
\STATE{$G$ := the labelled graph consisting only of a single vertex $v$ labelled $\langle A:\Q \rangle$ and no edges }
\STATE{$G$ := Communication\_{}Procedure($\T$,$G$,$v$)}
\STATE{$V$ := the label on the query vertex $\langle A:\Q \rangle$ in $G$}
\RETURN{$V$}
\end{algorithmic}
\end{algorithm}

Informally, the Communication Procedure can be explained as follows: The Query Minimization Procedure is called for $\T_A$ and $\alpha$. It returns a set of sets of \say-literals. For each such \say-literal set, we add a \say-literal set vertex connected to the root query vertex $\langle A:\alpha \rangle$ (lines 6-7). For each \say-literal in this set, we add a query vertex and an edge from the set vertex to this query vertex labelled by $\Tr$ or $\Fa$ depending on the sign of the \say-literal (8-15). We then apply Query Minimization Procedure and the rest of the procedure just explained to each new query vertex (line 22). At the same time, we label query vertices with truth values as follows: When all query vertices emerging from a \say-literal set vertex are labelled with the same truth value as the edge through which they are connected to the \say-literal set vertex, the query that produced that \say-literal set vertex is labelled $\Tr$ (lines 23-24). There is a dual procedure for labelling query vertices with $\Fa$ (lines 25-26). When a loop is detected, the query vertex causing the loop (by having the same label as a query vertex that is an ancestor of it) is labelled either with $\Fa$ or $\Un$, depending on whether the loop is over a negation (i.e. there is an $\Fa$-labelled edge in the path connecting the two vertices with the same label) or not (lines 16-20). $\Un$-labels  can also propagate towards the root of the graph (line 28).

\begin{algorithm}[h!t]
\caption{\textsf{Communication Procedure} \label{alg:comm}}
\begin{algorithmic}[1]
\REQUIRE{distributed theory $\T$, query graph $G$, query vertex $v$ of $G$,}
\ENSURE{updated query graph $G$}
\STATE{$k$ := the principal mentioned in the label of $v$}
\STATE{$\phi$ := the formula mentioned in the label of $v$}
\STATE{$\LL$ := \textsf{Query\_{}Minimization\_{}Procedure}($\T_k$,$\phi$)}
\WHILE{the input query vertex $v$ does not have a truth-value attached to it}
  \FOR{$L \in \LL$}
    \STATE{add a new \say-literal set vertex $L$ to $G$}
    \STATE{add to $G$ a new edge from vertex $v$ to vertex $L$}
      \FOR{$l \in L$}
        \STATE{$k'$ := the principal such that $l$ is of the form $k' \says \psi$ or $\neg k' \says \psi$}
        \STATE{$\psi$ := the formula such that $l$ is of the form $k' \says \psi$ or $\neg k' \says \psi$}
        \STATE{add a query vertex $v'$ labelled by $\langle k':\psi \rangle$ to $G$}
        \IF{$l$ is $k' \says \psi$}
          \STATE{add to $G$ a new edge labelled $\Tr$ from vertex $L$ to vertex $\langle k':\psi \rangle$}
        \ENDIF
        \IF{$l$ is $\neg k' \says \psi$}
          \STATE{add to $G$ a new edge labelled $\Fa$ from vertex $L$ to vertex $\langle k':\psi \rangle$}
        \ENDIF
        \IF{a query vertex $v''$ that is an ancestor of $v'$ is also labelled $\langle k':\psi \rangle$}
          \IF{all labelled edges between $v''$ and $v'$ are labelled by $\Tr$}
            \STATE{add $\Fa$-label to $v'$}
          \ELSE
	    \STATE{add $\Un$-label to $v'$}
          \ENDIF
        \ELSE
	  \STATE{Communication\_{}Procedure($\T$,$G$,$v'$)}
        \ENDIF
      \ENDFOR
      \IF{every query vertex $v'$ such that there is an edge from $L$ to $v'$ is labelled with the same truth value as this edge}
        \STATE{label $v$ with $\Tr$}
      \ENDIF
  \ENDFOR
  \IF{for every \say literal set vertex $L$ such that there is an edge from $v$ to $L$, there is a query vertex $v'$ such that there is an edge from $L$ to $v'$ labelled with the opposite truth value as $v'$}
    \STATE{label $v$ with $\Fa$}
  \ELSE
    \STATE{label $v$ with $\Un$}
  \ENDIF
\ENDWHILE
\RETURN{$G$}
\end{algorithmic}
\end{algorithm}

We continue Example \ref{ex:example} to illustrate the Communication Procedure.

\begin{example} Given the distributed theory $\T = \{T_A, T_B, T_C\}$, we query the principal $A$ about the truth value of $z$. We show the final graph in Figure \ref{fig:graph} and now explain its construction.
We start by calling the Communication Initialization Procedure; this generates a graph $G$ with the vertex $v =\langle A:z \rangle$  (with no associated truth value label). Then we call the Communication Procedure with arguments $\T$, $G$ and $v$, chich means that we must call the Query Minimization Procedure, which returns the set $\mathbb{L} = \{\{B \says p, B \says z\}, \{B \says r\}, \{\neg B \says r\}\}$ as shown in Example \ref{ex:query-min}. Since the input vertex $v$ has no truth-value associated to it, we next iterate over the sets $L \in \mathbb{L}$.
The \say-literal set vertex $\{B \says p, B \says z\}$ is added to $G$ with its corresponding edge.
Now we consider each \say-literal in the vertex. 

(i) For literal $B \says p$ generate a new query vertex $v^\prime = \langle B: p \rangle$ with an edge labelled $\mathbf{t}$ (since the literal is not negated) and recursively call the Communication Procedure with the updated graph as argument and vertex $v^\prime$. $v^\prime$ has no truth-value associated, the Query Minimization Procedure returns the set $\mathbb{L}^\prime = \{\{\}\}$; so after adding the \say-literal set vertex $\{\}$ connected to $v^\prime$, the truth-value $\mathbf{t}$ is assigned to $v^\prime$ by lines 21-22 of the Communication Procedure (this corresponds to the intuitive idea that $\mathbb{L}^\prime = \{\{\}\}$ means that $p$ is true in $T_B$).

(ii) For literal $B \says z$ generate a new query vertex $v^\prime = \langle B: z \rangle$ with an edge labelled $\mathbf{t}$ and recursively call the Communication Procedure with the updated graph as argument and vertex $v^\prime$. The Query Minimization Procedure is called returning the set $\mathbb{L} = \{\{C \says z\}\}$; for this literal we generate a new query vertex $v^{\prime\prime} = \langle C: z \rangle$ with an edge labelled $\mathbf{t}$. In turn, the Query Minimization Procedure is called returning the set $\mathbb{L}^{\prime\prime}= \{\{\neg B \says z\}\}$; for this literal we generate a new query vertex $v^{\prime\prime\prime} = \langle B: z \rangle$ with an edge labelled $\mathbf{f}$. 
At this point we detect a loop, as the query vertex $v^\prime$ that is an ancestor of $v^{\prime\prime\prime}$ is also labelled by $\langle B: z \rangle$. Since the loop contains an edge with label $\mathbf{f}$, the truth-value assignment for $v^{\prime\prime\prime}$ is  $\mathbf{u}$. This truth-value $\mathbf{u}$ is propagated up to label the query vertices $v^{\prime\prime}$ and $v^{\prime}$, since $\mathbf{u}$ does not match with neither $\mathbf{t}$ nor $\mathbf{f}$.

Finally, the truth-values for (i) matches the labeled edge, but not for the case of (ii). Thus we cannot yet label the root vertex with $\Tr$, and continue with the next \say-literal $\{B \says r\} \in \mathbb{L}$.

For vertex $\{B \says r\}$, we repeat the procedure as described above until we (again) detect a loop. This loop does not contain edges with label $\mathbf{f}$, so the truth-value assignment for the vertex at which the loop is detected is $\mathbf{f}$.
Again this truth-value is propagated to label the two query vertices above this vertex, as the labelled edges are labelled by $\mathbf{t}$. Since the truth-value $\Fa$ assigned to the query vertex $\langle B:r\rangle$ does not match the truth value of the labelled edge above it, the root vertex can still not be labelled with $\Tr$.

The subgraph produced below the final \say-literal set vertex $\{\neg B \says r\}$ is the same as below \say-literal set vertex $\{B \says r\}$, only that the labelled edge directly below this \say-literal set vertex is now labelled $\Fa$ instead of $\Tr$. So this time the label on the query vertex $\langle B:r\rangle$ matches the label on the labelled edge, to that the root vertex is labelled $\Tr$. This ends the main while loop and therefore the Communication Procedure. Finally, the Communication Procedure Initialization returns the output $\Tr$.

\begin{figure}[ht!] 
\centering
\begin{tikzpicture}[->,>=stealth', auto, node distance=2cm and 2cm, main node/.style={}]
  \node[main node] (Az) at (3,7) {$\langle A : z \rangle , \mathbf{t}$};
  \node[main node] (Bspz) at (0.5,6) {$\{B ~says~ p; B ~says~ z\}$};
  \node[main node] (Bsr) at (3.5,6) {$\{B ~says~ r\}$};
  \node[main node] (nBsr) at (6,6) {$\{\neg B ~says~ r\}$};
  \node[main node] (Bp) at (-0.5,5) {$\langle B : p \rangle , \mathbf{t}$};
  \node[main node] (Bsp) at (-0.5,4) {$\{\}$};
  \node[main node] (Bz1) at (1,5) {$\langle B : z \rangle ,  \mathbf{u}$};
  \node[main node] (Csz) at (1,4) {$\{C ~says~ z\}$};
  \node[main node] (Cz) at (1,3) {$\langle C : z \rangle , \mathbf{u}$};
  \node[main node] (Bsz) at (1,2) {$\{\neg B ~says~ z\}$};
  \node[main node] (Bz2) at (1,0.7) {\begin{tabular}{c} $\langle B : z \rangle ,  \mathbf{u}$ ~ loop! \\ (over negation)\end{tabular}};
  \node[main node] (Br1) at (3.5,5) {$\langle B : r \rangle ,  \mathbf{f} $};
  \node[main node] (Csr) at (3.5,4) {$\{C ~says~ r\}$};
  \node[main node] (Cr) at (3.5,3) {$\langle C : r \rangle , \mathbf{f} $};
  \node[main node] (Bsr2) at (3.5,2) {$\{B ~says~ r\}$};
  \node[main node] (Br2) at (3.5,0.7) {\begin{tabular}{c} $\langle B : r \rangle , \mathbf{f}$ ~ loop! \end{tabular}};
  \node[main node] (Br3) at (6,5) {$\langle B : r \rangle , \mathbf{f} $};
  \node[main node] (Csr2) at (6,4) {$\{C ~says~ r\}$};
  \node[main node] (Cr2) at (6,3) {$\langle C : r \rangle , \mathbf{f} $};
  \node[main node] (Bsr3) at (6,2) {$\{B ~says~ r\}$};
  \node[main node] (Br4) at (6,0.7) {\begin{tabular}{c} $\langle B : r \rangle , \mathbf{f}$ ~ loop! \end{tabular}};

  \path[every node/.style={font=\sffamily\normalsize},pos=0.5, auto]
    (Az) edge   [bend right=20] node [above] {} (Bspz)
    (Az) edge   [bend left=20]  node [above] {} (Bsr)
    (Az) edge   [bend left=20]  node [above] {} (nBsr)
    (Bspz) edge  node [right] {$\mathbf{t}$} (Bp)
    (Bp) edge    node [right] {} (Bsp)
    (Bspz) edge  node [right] {$\mathbf{t}$} (Bz1)
    (Bz1) edge   node [above] {} (Csz)
    (Csz) edge   node [right] {$\mathbf{t}$} (Cz)
    (Cz) edge    node [above] {} (Bsz)
    (Bsz) edge   node [right] {$\mathbf{f}$} (Bz2)
    (Bsr) edge   node [right] {$\mathbf{t}$} (Br1)
    (Br1) edge   node [above] {} (Csr)
    (Csr) edge   node [right] {$\mathbf{t}$} (Cr)
    (Cr) edge    node [above] {} (Bsr2)
    (Bsr2) edge  node [right] {$\mathbf{t}$} (Br2)
    (nBsr) edge  node [right] {$\mathbf{f}$} (Br3)
    (Br3) edge   node [above] {} (Csr2)
    (Csr2) edge  node [right] {$\mathbf{t}$} (Cr2)
    (Cr2) edge   node [above] {} (Bsr3)
    (Bsr3) edge  node [right] {$\mathbf{t}$} (Br4)
  ;
\end{tikzpicture}
\caption{Query Graph}
\label{fig:graph}
\end{figure}
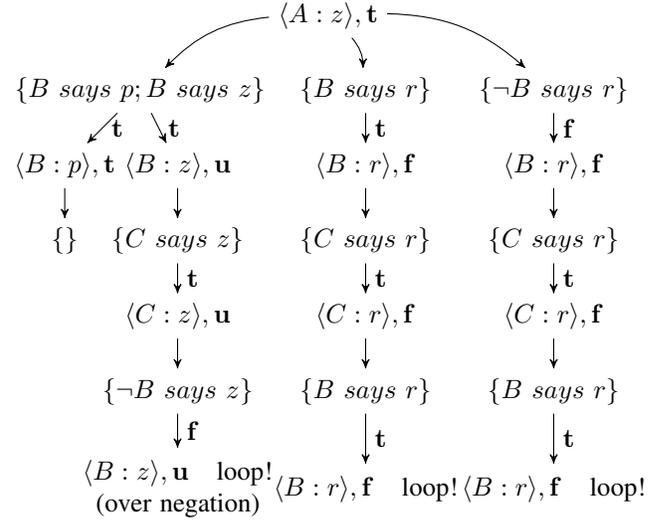

\end{example}

\subsection{Correctness of decision procedure}
\label{sec:relation}

The following theorem states that the result of the decision procedure is always in line with the well-founded semantics of \DACL:

\begin{theorem}
\label{thm:correspondence}
 Let $\T$ be a distributed theory, let $A$ be an agent, and let $\Q$ be a \DACL formula. When $A$'s theory $\T_A$ is queried about $\Q$, the decision procedure returns $(A \says \Q)^{\wfm(\T)}$, i.e.\ the truth value of $A \says \Q$ in the well-founded model of $\T$.
\end{theorem}

The proof of Theorem \ref{thm:correspondence} can be found in the Appendix.

\subsection{Complexity of the decision procedure} \label{sec:complexity}

The Query Minimization Procedure has a worst-case runtime that is exponential in the maximum of the number of different \say-atoms in $T$, the size of the vocabulary $\Sigma$ and the size of the domain $D$: Its \textbf{for}-loop has $3^n$ iterations, where $n$ is the number of different \say-atoms in $T$, and as min\_{}incons has a worst-case runtime exponential in the maximum of the size of the vocabulary $\Sigma$ and the size of the domain $D$. On the other hand, if we count each call to the Query Minimization Procedure as one step, the communication and loop-handling has runtime quasilinear in the number of subqueries called.

To make the decision procedure practically applicable, heuristics would have to employed to reduce the runtime for determining an access right, and a principled approach for dealing with situations when access cannot be determined within a reasonable amount of time would be required (see Section VII of Cramer et al.\ \cite{NP-complete} for an example of such an approach in a somewhat different access-control setting). One modification of the decision procedure that reduces the expected runtime, even though it does not reduce the worst-case runtime, is to not calculate the whole of $\mathbb{L}$ immediately in the Query Minimization Procedure, but to instead first calculate just one $L \in \mathbb{L}$, then do the communication necessary for determining whether this $L$ actually makes the query true, and continue with the step-wise calculation of $\mathbb{L}$ only if the query has not yet been determined true. 

\section{Related Work}
\label{sec:related}
Most access control logics proposed in the literature have been defined in a proof-theoretical way, i.e. by specifying which axioms and inference rules they satisfy. This contrasts with Van Hertum et al.'s \cite{ijcai16Cramer} approach of defining \DACL semantically rather than proof-theoretically. This difference means that the tasks of defining decision procedures for these access control logics involve very different technical machinery.


Garg and Abadi \cite{Garg2008,Garg08principal-centricreasoning} and Genovese \cite{Genovese12} have defined Kripke semantics for many of the access control logics that were previously defined proof-theoretically in the literature. They introduced these Kripke semantics as a tool for defining decision procedures for those access control logics. Genovese \cite{Genovese12} follows the methodology of Negri and von Plato \cite{DBLP:books/daglib/0005072,DBLP:books/daglib/0036978} of using a Kripke semantics of a modal logic to define Labelled Sequent Calculus, which forms the basis of a decision procedure for the logic.

Denecker et al.\ \cite{DeneckerMT03} have defined a procedure for computing the well-founded model of an autoepistemic theory. This procedure might be extendable to a procedure for computing the well-founded model of \DACL. However, such an extension of their procedure would not have the feature of minimizing the communication between principals, and thus violate the need-to-know principle and cause privacy concerns (see Section \ref{sec:minimization_motivation}).


\section{Conclusion and Future Work}
\label{sec:future}
In this paper, we have defined a query-based decision procedure for the well-founded semantics of \DACL. When applying \DACL to access control, this decision procedure allows one to determine access rights while avoiding redundant information flow between principals in order to enhance security and reduce privacy concerns.

Given that our decision procedure has in the worst case an exponential runtime (see Section \ref{sec:complexity}), a more efficient decision procedure will have to be developed for \DACL or an expressively rich subset of it in order to apply it in practice. 
For this reason, the contribution of this paper is mainly of a conceptual nature: The defined decision procedure is a proof of concept that increases our understanding of dAEL(ID) by providing an algorithmic characterization of the well-founded semantics of dAEL(ID). This algorithmic characterization complements in a conceptually fruitful way the semantic definition from Van Hertum et al.\ \cite{ijcai16Cramer} which is based on a fixpoint construction on abstract structures.

Our decision procedure aims at proving a query in terms of queries to other principals. In this process, it cautiously handles possible loops between such queries. This is highly reminiscent of the way \emph{justifications} are defined, for instance for logic programs \cite{lpnmr/DeneckerBS15}. Hence it may be interesting to define justification semantics for \DACL.



\bibliographystyle{IEEEtran}
\bibliography{references}

\appendix

\section{Proof of Theorem 1}

The following lemma states that the Query Minimization Procedure (Algorithm \ref{alg:query}) really does what it is supposed to do:

\begin{lemma}
\label{lem:minimization}
 Let $T$ be a \DACL theory and let $\alpha$ be a \DACL formula. The set $\mathbb{L}$ returned by \textsf{Query\_{}Minimization\_{}Procedure}$(T,\alpha)$ is
 \begin{align*}
 \{L \;|\; &L \text{ is minimal (under set inclusion) among the sets $L'$ of}\\
 &\text{\say-literals that make $\alpha$ true with respect to $T$}\}
 \end{align*}
\end{lemma}

%
%
%

Let $\bot$ denote the DPWS in which every agent's possible world structure is the set of all structures over domain $D$ and vocabulary $\Sigma$. Let $\top$ denote the DPWS in which every agent's possible world structure is the empty set.

The well-founded model of $\T$ is the $\leq_p$-least fixpoint of $S_\T$. When the domain is finite, as we are assuming when applying the decision procedure, there is a natural number $n$ such that $\wfm(\T) = (S_\T)^n(\bot,\top)$. In other words, the well-founded model can be computed by a finite number of application of $S_\T$ to $(\bot,\top)$, until a fixpoint is reached.

The steps in the decision procedure defined in section \ref{sec:query} do not directly correspond to the steps in the computation of the well-founded model by a finite number of application of $S_\T$ to $(\bot,\top)$. In order to prove that the two computations nevertheless always yield the same result, we first define a decision procedure that resembles the decision procedure defined in section \ref{sec:query}, but whose steps correspond more directly to the iterative application of $S_\T$ to $(\bot,\top)$. We call this auxiliary decision procedure the \emph{$S_\T$-based decision procedure}. So we prove Theorem \ref{thm:correspondence} by proving two things:
\begin{itemize}
 \item The decision procedure defined in section \ref{sec:query} is equivalent to the $S_\T$-based decision procedure.
 \item  When $A$'s theory $\T_A$ is queried about $\Q$, the $S_\T$-based decision procedure returns \textsf{yes} iff $(A \says \Q)^{\wfm(\T)} = \Tr$.
\end{itemize}


In the definition of the $S_\T$-based decision procedure, we use a \emph{query graph} as defined in section \ref{sec:query}. 


There is a direct correspondence between distributed belief pairs and certain truth-value labelling of the query vertices in a query graph:
\begin{definition}
 Let $G$ be a query graph. Let \ubp be a distributed belief pair. We say that the truth-value labelling of the query vertices of $G$ \emph{corresponds to} \ubp iff for each query vertex $k:\phi$ in $G$, the truth-value with which this vertex is labelled is $(k \says \phi)^\ubp$.
\end{definition}

Note that there are truth-labellings of the query vertices that do not correspond to any distributed belief pair. We call a truth-labelling of the query vertices \emph{good} iff it corresponds to some distributed belief pair.

The $S_\T$-based decision procedure works by first producing a query graph and then iteratively modifying the truth-value labelling of the query vertices. We need to ensure that after each iteration of this iterative modification, the truth-value labelling of the query vertices is good. However, there are intermediate steps within each iteration which lead to a bad labelling of the query vertices. In order to get back to a good labelling, we apply the changes defined by Algorithm \ref{alg:make good}.

\begin{algorithm}[h!t]
\caption{\textsf{Make labelling of query vertices good} \label{alg:make good}}
\begin{algorithmic}[1]
\REQUIRE{query graph $G$}
\ENSURE{modified query graph $G$}
\WHILE{there is a $\Un$-labelled query vertex $k:\phi$ in $G$ such that replacing \say-atoms in $\phi$ corresponding to $\Tr$- or $\Fa$-labelled query vertices by $\Tr$ and $\Fa$ respectively makes $\phi$ a tautology}
  \STATE{change the $\Un$-label in each such query vertex in $G$ by $\Tr$}
\ENDWHILE
\WHILE{there is a $\Fa$-labelled query vertex $k:\phi$ in $G$ such that replacing \say-atoms in $\phi$ corresponding to $\Un$-, $\Tr$- or $\Fa$-labelled query vertices by $\Tr$ or $\Fa$, $\Tr$ and $\Fa$ respectively makes $\phi$ a tautology}
  \STATE{change the $\Fa$-label in each such query vertex in $G$ by $\Un$}
\ENDWHILE
\RETURN{$G$}
\end{algorithmic}
\end{algorithm}


In order to define the $S_\T$-based decision procedure, we furthermore need the following two definitions:

\begin{definition}
In a query graph, a \say-literal set vertex $L$ is defined to be \emph{satisfied} if for every $\Tr$-labelled edge from $L$ to a query vertex, the query vertex is labelled by $\Tr$, and for every $\Fa$-labelled edge from $L$ to a query vertex, the query vertex is labelled $\Fa$.
\end{definition}


\begin{definition}
In a query graph, a \say-literal set vertex $L$ is defined to be \emph{dissatisfied} if either for some $\Tr$-labelled edge from $L$ to a query vertex, the query vertex is labelled by $\Fa$, or for some $\Fa$-labelled edge from $L$ to a query vertex, the query vertex is labelled $\Tr$.
\end{definition}

The definition of the $S_\T$-based decision procedure is given by the pseudo-code under Algorithm \ref{alg:ST-based}. 

\begin{algorithm}[h!t]
\caption{\textsf{$S_\T$-based decision procedure} \label{alg:ST-based}}
\begin{algorithmic}[1]
\REQUIRE{distributed theory $\T$, principal $A$, \DACL formula $\Q$}
\ENSURE{truth-value $V \in \{\Tr,\Fa,\Un\}$}
\STATE{$G$ := the empty graph }
\STATE{add a new query vertex $A:\Q$ to $G$}
\STATE{query\_{}stack := $\langle A:\Q \rangle$}
\WHILE{query\_{}stack $\neq \langle \rangle$}
  \STATE{$k:\phi$ := first element of query\_{}stack}
  \STATE{$\LL$ := \textsf{Query\_{}Minimization\_{}Procedure}($\T_k$,$\phi$)}
  \FOR{$L \in \LL$}
    \IF{$G$ does not contain a \say-literal set vertex $L$}
      \STATE{add a new \say-literal set vertex $L$ to $G$}
      \FOR{$l \in L$}
        \STATE{$k'$ := the principal such that $l$ is of the form $k' \says \psi$ or $\neg k' \says \psi$}
        \STATE{$\psi$ := the formula such that $l$ is of the form $k' \says \psi$ or $\neg k' \says \psi$}
        \IF{$G$ does not contain a query vertex $k':\psi$}
           \STATE{add a query vertex $k':\psi$ to $G$}
           \STATE{add $k':\psi$ to query\_{}stack}
        \ENDIF
        \IF{$l$ is $k' \says \psi$}
          \STATE{add to $G$ a new edge labelled $\Tr$ from vertex $L$ to vertex $k':\psi$}
        \ENDIF
        \IF{$l$ is $\neg k' \says \psi$}
          \STATE{add to $G$ a new edge labelled $\Fa$ from vertex $L$ to vertex $k':\psi$}
        \ENDIF
      \ENDFOR
    \ENDIF
    \STATE{add to $G$ a new edge from vertex $k:\phi$ to vertex $L$}
  \ENDFOR
\ENDWHILE
\STATE{add the label $\Un$ to all query vertices in $G$}
\STATE{finished := 0}
\WHILE{finished = 0}
  \STATE{$G_1 := G$}
  \STATE{change every $\Tr$-label on a query vertex in $G_1$ to $\Un$}
  \WHILE{in $G_1$ there is a query vertex labelled by $\Un$ with an edge to a satisfied \say-literal set vertex}
    \STATE{change every $\Un$-label on a query vertex with an edge to a satisfied \say-literal set vertex to $\Tr$}
    \STATE{$G$ := \textsf{Make\_{}labelling\_{}of\_{}query\_{}vertices\_{}good($G$)}}
  \ENDWHILE
  \STATE{$G_2 := G$}
  \STATE{change every $\Fa$-label on a query vertex in $G_2$ to $\Un$}
  \WHILE{in $G_2$ there is a query vertex labelled by $\Un$ with an edge to a dissatisfied \say-literal set vertex}
    \STATE{change every $\Un$-label on a query vertex with an edge to a dissatisfied \say-literal set vertex to $\Fa$}
    \STATE{$G$ := \textsf{Make\_{}labelling\_{}of\_{}query\_{}vertices\_{}good($G$)}}
  \ENDWHILE
  \STATE{in $G$, change the label on all query vertices that are labelled $\Un$ in $G$ and labelled $\Tr$ in $G_1$ into $\Tr$}
  \STATE{in $G$, change the label on all query vertices that are labelled $\Un$ in $G$ and labelled $\Fa$ in $G_2$ into $\Fa$}
  \IF{no changes were made to $G$ in the previous two lines}
    \STATE{finished := 1}
  \ENDIF
\ENDWHILE
\STATE{$V$ := the label on the query vertex $A:\Q$ in $G$}
\RETURN{$V$}
\end{algorithmic}
\end{algorithm}

We now sketch the proof of the equivalence between the $S_\T$-based decision procedure and the decision procedure defined in section \ref{sec:query}: The only fundamental difference between these two decision procedures is the loop-handling. Step 2) of the $S_\T$-based decision procedure takes care of making queries looping over $\Tr$-labelled edges false. Queries looping over $\Fa$-labelled edges will always be left undecided by the $S_\T$-based decision procedure, which corresponds to making them undecided in the decision procedure defined in section \ref{sec:query}.

We now establish that the $S_\T$-based decision procedure always gives the same result as the well-founded semantics. Note that the labelling corresponding to the distributed belief pair $(\bot,\top)$ is the labelling in which all query vertices are labelled by $\Un$. Keeping in mind that the well-founded model can be computed by a finite number of application of $S_\T$ to $(\bot,\top)$, it is now easy to see that the following lemma is sufficient to establish that the $S_\T$-based decision procedure always gives the same result as the well-founded semantics:

\begin{lemma}
Let $\T$ be a distributed theory, $A$ be a principal and $\Q$ be a \DACL formula. Let $G$ be the query graph produced by lines 1-20 of Algorithm \ref{alg:ST-based} applied to $\T$, $A$ and $\Q$. Let \ubp be a distributed belief pair. Labelling the query vertices in $G$ according to \ubp and then applying lines 24 to 34 of Algorithm \ref{alg:ST-based} to $G$ yields a labelling of the queries corresponding to $S_\T(B)$.
\end{lemma}

\begin{proof}
For proving this lemma, it is enough to prove the following four properties, which can be proved straightforwardly:
\begin{enumerate}
 \item The change in the truth-value labelling of the query vertices of $G_1$ in line 25 of Algorithm \ref{alg:ST-based} corresponds to changing the belief pair $(\upws_1,\upws_2)$ to $\bot,\upws_2)$. 
 \item The change in the truth-value labelling of the query vertices of $G_1$ in lines 27-28 of Algorithm \ref{alg:ST-based} corresponds to changing the belief pair $(\upws_1,\upws_2)$ to $(D^*_\T(\upws_1,\upws_2)_1,\upws_2)$.
 \item The change in the truth-value labelling of the query vertices of $G_2$ in line 30 of Algorithm \ref{alg:ST-based} corresponds to changing the belief pair $(\upws_1,\upws_2)$ to $(\upws_1,\top)$.
 \item The change in the truth-value labelling of the query vertices of $G_2$ in lines 32-33 of Algorithm \ref{alg:ST-based} corresponds to changing the belief pair $(\upws_1,\upws_2)$ to $(\upws_1,D^*_\T(\upws_1,\upws_2)_2)$.
 \item Let $\upws_1,\upws_2,\upws_3,\upws_4$ be DPWS's such that $\upws_3 \leq_K \upws_1 \leq_K \upws_4 \leq_K \upws_2$. If the query vertices are labelled $\Tr$ in correspondence with the distributed belief pair $(\upws_1,\upws_2)$, labelled $\Fa$ in correspondence with the distributed belief pair $(\upws_3,\upws_4)$, and labelled $\Un$ otherwise, the resulting labelling corresponds to the distributed belief pair $(\upws_1,\upws_4)$.
\end{enumerate}

\end{proof}

This completes the proof of Theorem \ref{thm:correspondence}.

\end{document}